\documentclass[10pt,conference]{IEEEtran}
\usepackage{amsmath,amssymb}
\usepackage{graphicx}
\usepackage{mathrsfs}
\usepackage{tikz}
\usepackage{booktabs}
\usepackage{enumerate}
\usepackage{psfrag}	
\usepackage{array}
\usepackage{multirow,hhline}
\usepackage{exscale}
\usepackage{color}
\usepackage{colortbl}
\usepackage{epsfig}
\usepackage{cite}
\usepackage{amsthm}

\newcommand{\inr}{\ensuremath{\mathsf{INR}}}
\newtheorem{theorem}{Theorem}
\newtheorem{lemma}{Lemma}

\newtheorem{definition}{Definition}

\begin{document}

\IEEEoverridecommandlockouts
\title{On Gaussian Multiple Access Channels with Interference: Achievable Rates and Upper Bounds}
\author{
\IEEEauthorblockN{Anas Chaaban, Aydin Sezgin}
\IEEEauthorblockA{Emmy-Noether Research Group on Wireless Networks\\
Institute of Telecommunications and Applied Information Theory\\
Ulm University, 89081 Ulm, Germany\\
Email: {anas.chaaban@uni-ulm.de, aydin.sezgin@uni-ulm.de}}
\and
\IEEEauthorblockN{Bernd Bandemer, Arogyaswami Paulraj}
\IEEEauthorblockA{Information Systems Lab\\
Stanford University, Packard Building,\\
350 Serra Mall, Stanford, CA 94305-9510, U.S.A.\\
Email: {bandemer@stanford.edu, apaulraj@stanford.edu}}
\thanks{%
The work of A. Chaaban and A. Sezgin is supported by the German Research Foundation, Deutsche Forschungsgemeinschaft (DFG), Germany, under grant SE 1697/3. The work of B. Bandemer is supported by an Eric and Illeana Benhamou Stanford Graduate Fellowship.%
}
}

\maketitle


\begin{abstract}
We study the interaction between two interfering Gaussian 2-user multiple access channels. The capacity region is characterized under mixed strong--extremely strong interference and individually very strong interference. Furthermore, the sum capacity is derived under a less restricting definition of very strong interference. Finally, a general upper bound on the sum capacity is provided, which is nearly tight for weak cross links.
\end{abstract}

\begin{IEEEkeywords}
Gaussian MAC, capacity, bounds, strong interference, very strong interference. 
\end{IEEEkeywords}

\section{Introduction}
A scenario where several transmitters each want to deliver a message to a common receiver is known as the multiple access channel (MAC). This setup models mobile users  that want to communicate with a central base station in a cellular network, for example. The MAC capacity region is known since 1971 \cite{Ahlswede,Liao}.

Another intensively studied model in information theory is the interference channel (IC). In this model, two transmit-receive pairs want to communicate while causing interference to each other. First proposed in 1978 \cite{Carleial}, the interference channel is still not fully understood. Its capacity  is known only in special cases, e.g., the very-strong interference regime \cite{Carleial_vsi}, the strong interference regime \cite{Sato}, and the noisy interference regime \cite{AnnapureddyVeeravalli,ShangKramerChen,MotahariKhandani} where only its sum-capacity is known. The sum-capacity of the interference channel with mixed interference was analyzed in \cite{WengTuninetti}.

The MAC and the IC are the two building blocks of the model considered here. We consider a setup that models two interfering 2-user MACs. This is a very practical situation which occurs frequently in cellular networks, where multiple mobile stations communicate with the base stations in their respective cells. The degrees of freedom of this setup were studied in \cite{CadambeJafarWang} and \cite{SuhTse}. We follow the naming in \cite{SuhTse} where the interfering multiple access channel was called the IMAC. We study this model and obtain new capacity results.

The capacity region of the IMAC is derived for a case of mixed strong--extremely strong interference. That is, when at each receiver, one interferer satisfies a strong interference condition and the other interferer satisfies an extremely strong interference condition. In this case, we show that the capacity region of the IMAC is bounded by the capacity region of the MAC formed by the two desired signals and the strong interfering signal at each receiver. This region is achievable by using Gaussian codes, decoding the extremely strong interferer first and subtracting it from the received signal, and then using the capacity achieving scheme for the resulting MAC to decode the remaining three signals. 

A condition for individually very strong interference is derived, and when this condition is satisfied, interference does not decrease the \emph{capacity region} of each of the interfering MACs, i.e., their interference free \emph{capacity region} can be achieved. Furthermore, another condition is derived (very strong combined interference), under which interference does not decrease the \emph{sum capacity} of each of the interfering MACs, i.e., their interference-free \emph{sum capacity} can be achieved. 

The simple scheme of treating interference as noise at each receiver gives a sum capacity lower bound for the IMAC. Using a genie aided approach similar to \cite{AnnapureddyVeeravalli}, we obtain a sum capacity upper bound which, although not coinciding with the lower bound of treating interference as noise, is fairly tight if the interference power is low.

\section{System Model}
\label{Model}

We consider the \emph{interfering MAC (IMAC)} channel depicted in Figure \ref{IMAC}, in which two 2-user multiple access channels use the same transmission resource and therefore interfere with each other. In this channel, transmitters 1 and 2 would like to send independent messages to receiver 1, while transmitters 3 and 4 have independent messages for receiver 2. Each of the two receiver nodes observes the combination of two desired and two interfering signals.

We constrain our attention to the symmetric real-valued memoryless Gaussian setting, where the channel inputs are real numbers, the observation noise is additive Gaussian, and at each time instance, the channel outputs are given by
\begin{align}
Y_1= X_1+ X_2+ h_{1}X_3+h_{2} X_4 + Z_1,\\
Y_2= h_1 X_1 + h_2 X_2 + X_3 + X_4 + Z_2.
\end{align}
Here, $h_1$ and $h_2$ denote the channel coefficients of the undesired cross-links. The noise terms $Z_1$, $Z_2$ are independent unit variance Gaussian  random variables. The channel inputs $X_i$ are controlled by the corresponding transmit node $i$, and are subject to the average power constraints 
\begin{align*}
	\mathbb{E}[X_1^2]&\leq P_1,  & 	\mathbb{E}[X_3^2]&\leq P_3 = P_1,  \\
	\mathbb{E}[X_2^2]&\leq P_2,  & 	\mathbb{E}[X_4^2]&\leq P_4 = P_2. 
\end{align*}
The channel is therefore completely symmetric with respect to exchanging the two multiple-access channels. It is parameterized by the tuple $(P_1,P_2,h_1,h_2)$.

\begin{figure}[h]
\centering
\includegraphics[width=0.8\columnwidth]{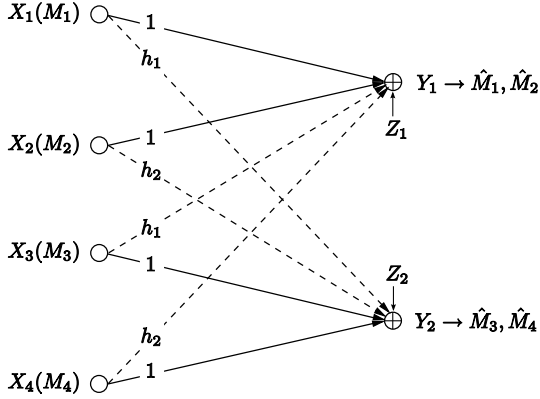}
\caption{Multiple access channel with interference. Users 1 and 2 want to communicate with receiver 1, while users 3 and 4 want to communicate with receiver 2.}
\label{IMAC}
\end{figure}

A $(n,2^{nR_1},2^{nR_2},2^{nR_3},2^{nR_4})$ code for the IMAC consists of message sets $\mathcal{M}_i=\{1,\dots,2^{nR_i}\}$, four encoding functions $f_i: \mathcal M_i \to \mathbb R^n$ (one for each transmitting node), and four decoding functions $g_i : \mathbb R^n \to \mathcal M_i$ (two for each receiving node). The probability of decoding error for the code $c$ is
\begin{align}
	P_e^{(n)} &= P (\hat M_i \neq M_i \text{ for some $i$ } ),
\end{align}
where the messages $M_i$ are uniformly and independently drawn from the message sets, and $\hat M_i$ are the detected messages at the receivers, resulting from applying the decoding functions $g_i$.

A rate tuple $(R_1,R_2,R_3,R_4)$ is achievable in the IMAC if there exists a sequence of $(n,2^{nR_1},2^{nR_2},2^{nR_3},2^{nR_4})$ codes such that $P_e^{(n)} \to 0$ as $n \to \infty$. The capacity region $\mathcal C$ of the IMAC is the closure of the set of all achievable rate tuples. The sum capacity is the largest achievable sum-rate 
\begin{align}
C_\Sigma=\max_{(R_1,R_2,R_3,R_4)\in\mathcal{C}} R_1+R_2+R_3+R_4.
\end{align}

\section{Main Results}
\label{MainResults}

Before we state the main results of this paper, we need the following definition for simplicity of exposition. Consider the MAC channels that are contained in the IMAC.
\begin{definition}
Let $M(\mathcal{S},j,N)$ denote the multiple access channel (MAC) from transmitters $i\in\mathcal{S}\subseteq\{1,2,3,4\}$ to receiver $j\in\{1,2\}$ with additive Gaussian noise of variance $N$. Let $\mathcal{C}^M(\mathcal{S},j,N)$ be the capacity region of this MAC.
\end{definition}

It is well-known~\cite{CoverThomas} that $\mathcal{C}^M(\mathcal{S},j,N) \subseteq \mathbb R_+^{\lvert \mathcal S \rvert}$ is specified by the inequalities
\begin{align*}
\sum_{i\in\mathcal{T}}R_i\leq\frac{1}{2}\log\left(1+\sum_{i\in\mathcal{T}}h_{ij}^2P_i/N\right), \quad \forall \mathcal{T}\subseteq\mathcal{S},
\end{align*}
where $h_{ij} \in \{1,h_1,h_2\}$ is the channel coefficient from the $i$th transmitter to the $j$th receiver.

When convenient, we abbreviate $M(\mathcal{S},j,1)$ as  $M(\mathcal{S},j)$ and $\mathcal{C}^M(\mathcal{S},j,1)$ as $\mathcal{C}^M(\mathcal{S},j)$.  The following lemma  will be useful in the proof of subsequent theorems.

\begin{lemma}
\label{SILemma}
The capacity of the IMAC $\mathcal{C}^{}$ is included in $\overline{\mathcal{C}}^{}$, i.e.$$\mathcal{C}^{}\subseteq\overline{\mathcal{C}}^{}$$
where $\overline{\mathcal{C}}^{}=$
\begin{align}
\left\{
\begin{array}{l}
R_i\geq0:\\
(R_1,R_2)\in\mathcal{C}^M(\{1,2\},1)\\
(R_3,R_4)\in\mathcal{C}^M(\{3,4\},2)\\
(R_1,R_2,R_3)\in\mathcal{C}^{M}(\{1,2,3\},1)  \text{ if } h_{1}^2\geq1\\
(R_1,R_2,R_4)\in\mathcal{C}^{M}(\{1,2,4\},1) \text{ if } h_{2}^2\geq1\\
(R_1,R_3,R_4)\in\mathcal{C}^{M}(\{1,3,4\},2) \text{ if } h_{1}^2\geq1\\
(R_2,R_3,R_4)\in\mathcal{C}^{M}(\{2,3,4\},2) \text{ if } h_{2}^2\geq1
\end{array}\right\}
\end{align}
\end{lemma}
\begin{proof}[Proof Sketch]
The bound on $(R_1,R_2)$ is trivial since $M(\{1,2\},1)$ is the interference-free version of the first MAC, and the presense of interference cannot improve the rates. Similarly $(R_3,R_4)\in\mathcal{C}^M(\{3,4\},2)$. For the other bounds, assume that $h_1^2 \geq 1$. If we give $M_4$ to the first receiver as genie information, it can construct a less noisy version of $Y_2^n$ from its own observation $Y_1^n$. Therefore, $M_3$ can then be reliably decoded from $Y_1^n$, i.e., $(R_1,R_2,R_3)\in\mathcal{C}^{M}(\{1,2,3\},1)$. The other bounds are obtained similarly.
\end{proof}

\subsection{Capacity with mixed strong--extremely strong interference}
Consider the following special case of the IMAC.

\begin{definition}
The IMAC has mixed strong--extremely strong interference MSES$(i,j)$ if for $i,j\in\{1,2\}$, $i\neq j$, we have
\begin{align}
\label{MSVS1}
h_{j}^2&\geq 1+P_1+P_2+h_{i}^2P_i,\\
\label{MSVS4}
h_{i}^2&\geq1.
\end{align}
where $h_i$ represents the strong interference channel, and $h_j$ the extremely strong one.
\end{definition}

The capacity region then follows from this theorem.

\begin{theorem}
\label{VSI-SI-Region}
The capacity region of the IMAC with mixed strong--extremely strong interference MSES$(1,2)$ is given by 
\begin{equation}
\mathcal{C}^{}=\left\{
\begin{array}{l}
(R_1,R_2,R_3,R_4)\in\mathbb{R}^4_+:\\
(R_1,R_2,R_3)\in\mathcal{C}^M(\{1,2,3\},1)\\
(R_1,R_3,R_4)\in\mathcal{C}^M(\{1,3,4\},2)
\end{array}\right\}
\end{equation}
\end{theorem}

Note that due to  symmetry in the channel, the sets $\mathcal{C}^M(\{1,2,3\},1)$ and $\mathcal{C}^M(\{1,3,4\},2)$ are in fact equal. A similar result holds for the other case MSES$(2,1)$.

\begin{proof}[Proof Sketch]
The outer bound is obtained from Lemma \ref{SILemma}. The inner bound is obtained using the following scheme. Receivers decode the extremely strong interfering signal first while treating all other signals as noise. That is, $X_4^n$ is decoded first at receiver 1 while treating $X_1^n$, $X_2^n$, and $X_3^n$ as noise, and  $X_1^n$ is decoded first at receiver 2 while treating $X_2^n$, $X_3^n$, and $X_4^n$ as noise. This is reliably possible due to condition~\eqref{MSVS1}. Then the receivers remove the contribution of the decoded interference from their received signal, and decode the remaining signals in a MAC fashion, achieving the outer bound.
\end{proof}

\subsection{Capacity with individually very strong interference}
Inspired by the interference channel with very strong interference \cite{Carleial_vsi}, where the presence of cross-links does not impair the capacity region, we now consider the following special case of the IMAC.

\begin{definition}
The IMAC has individually very strong interference if
\begin{align}
\label{VSI-1}
h_{1}^2,h_2^2&\geq 1+P_1+P_2,
\end{align}
\end{definition}

We call this \emph{individually very strong}, since both cross-link gains have to satisfy separate conditions.

\begin{theorem}
\label{VSI-Region}
The capacity region of  the IMAC with individually very strong interference is 
\begin{equation}
\mathcal{C}^{}=\left\{
\begin{array}{l}
R_i\geq0:\\
(R_1,R_2)\in\mathcal{C}^M(\{1,2\},1),\\
(R_3,R_4)\in\mathcal{C}^M(\{3,4\},2),
\end{array}
\right\}
\end{equation}
\end{theorem}

As in the case of the interference channel, the capacity region is not impaired by the presence of cross-links, i.e., the interference-free capacity is achieved. Note that because of symmetry in the channel, the sets $\mathcal{C}^M(\{1,2\},1)$ and $\mathcal{C}^M(\{3,4\},2)$ are in fact equal.

\begin{proof}[Proof Sketch]
The outer bound is given by Lemma \ref{SILemma}. This outer bound is achievable as follows. Transmitters use Gaussian codebooks. Each receiver decodes both interfering signals first while treating both desired signals as noise. Reliable decoding of interference is possible if $(R_1,R_2)\in\mathcal{C}^M(\{1,2\},2,1+P_1+P_2)$ and $(R_3,R_4)\in\mathcal{C}^M(\{3,4\},1,1+P_1+P_2)$. Then, each receiver subtracts the contribution of the interfering signals, and decodes the desired signals interference free. Reliable decoding of the desired signals is possible if $(R_1,R_2)\in\mathcal{C}^M(\{1,2\},1)$ and $(R_3,R_4)\in\mathcal{C}^M(\{3,4\},2)$. Now if condition (\ref{VSI-1}) holds, then $\mathcal{C}^M(\{1,2\},1)\subseteq\mathcal{C}^M(\{1,2\},2,1+P_1+P_2)$ and $\mathcal{C}^M(\{3,4\},2)\subseteq\mathcal{C}^M(\{3,4\},1,1+P_1+P_2)$ and hence the regions $\mathcal{C}^M(\{1,2\},1)$ and $\mathcal{C}^M(\{3,4\},2)$ are achievable.
\end{proof}

As shown in Figure \ref{MAC-Regs}, condition (\ref{VSI-1}) guarantees that the region $\mathcal{C}^M(\{3,4\},2)$ (solid blue) is completely contained in  $\mathcal{C}^M(\{3,4\},1,1+P_1+P_2)$ (dashed red). The intuition is that the first receiver can decode the messages from transmitters $3$ and $4$ even under the additional noise caused by the first two transmitters.

\begin{figure}
\centering
\psfragscanon
\psfrag{If}[l]{{$\mathcal{C}^M(\{3,4\},2)$}}
\psfrag{C1}[l]{{$\mathcal{C}^M_{IVS}(\{3,4\},1,1+P_1+P_2)$}}
\psfrag{C2}[l]{{$\mathcal{C}^M_{VSC}(\{3,4\},1,1+P_1+P_2)$}}
\psfrag{x}[l]{$R_3$}
\psfrag{y}[l]{$R_4$}
\includegraphics[width=0.9\columnwidth]{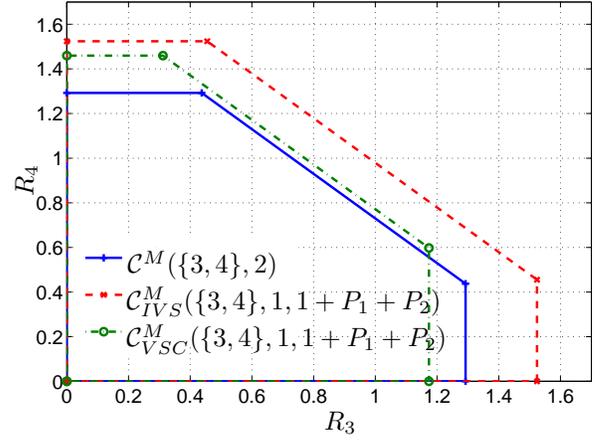}
\caption{The interference-free MAC $M(\{3,4\},2)$ as compared to $M(\{3,4\},1,1+P_1+P_2)$ under the scenarios of Theorem~\ref{VSI-Region} (individually very strong, IVS) and Theorem~\ref{VSI-Sum-Rate} (very strong combined, VSC).}
\label{MAC-Regs}
\end{figure}

\subsection{Sum capacity with very strong combined interference}
Now consider a weaker condition than~\eqref{VSI-1}.

\begin{definition}
The IMAC has  very strong combined interference if
\begin{align}
\label{VSI-5}
h_{1}^2P_1+h_{2}^2P_2&\geq (P_1+P_2)(1+P_1+P_2),
\end{align}
\end{definition}

We call this the very strong combined interference regime because the condition is on the sum of the interference powers at each receiver. It is clear that individually very strong interference implies very strong combined interference. The converse, however, does not hold.

This special case permits the following result.

\begin{theorem}
\label{VSI-Sum-Rate}
The sum capacity of the IMAC with  very strong combined interference is
\begin{align}
C_\Sigma&=\log(1+P_1+P_2).
\end{align}
\end{theorem}

This means that the sum capacities of the interference-free MACs $M(\{1,2\},1)$ and $M(\{3,4\},2)$, namely $1/2 \cdot \log(1+P_1+P_2)$ each, are achievable in the IMAC.

\begin{proof}[Proof Sketch]
From the outer bound in Lemma \ref{SILemma}, we know that the sum capacity is upper bounded by $\log(1+P_1+P_2)$. Now, by using the scheme in the proof of Theorem \ref{VSI-Region}, we show that the following region is achievable
\begin{align*}
&\underline{\mathcal{C}}=\\
&\left\{\begin{array}{l}
(R_1,R_2,R_3,R_4)\in\mathbb{R}_+^4:\\
(R_1,R_2)\in\mathcal{C}^M(\{1,2\},1)\cap\mathcal{C}^M(\{1,2\},2,1+P_1+P_2)\\
(R_3,R_4)\in\mathcal{C}^M(\{3,4\},2)\cap\mathcal{C}^M(\{3,4\},1,1+P_1+P_2)\end{array}\right\}
\end{align*}
Under condition (\ref{VSI-5}), the maximum of the sum of the achievable rates in $\underline{\mathcal{C}}$ is $\log(1+P_1+P_2)$ which is equal to the upper bound.
\end{proof}

An example is shown in Figure~\ref{MAC-Regs}. Although $\mathcal{C}^M(\{3,4\},1,1+P_1+P_2)$ (dot-dashed green) is not a superset of $\mathcal{C}^M(\{3,4\},2)$ (solid blue), it still does not constrain the sum rate to a value below the sum capacity of $\mathcal{C}^M(\{3,4\},2)$. Therefore, the overall sum rate of the IMAC is not impaired by the presence of cross-links. 

We note in passing that the achievable rate region $\underline{\mathcal{C}}$ in the proof above is a cartesian product, i.e., it does not contain conditions that couple the two constituent MACs. We therefore do not expect this region to be optimal in general.

Figure \ref{Thresh} shows the channel parameter range where the \emph{capacity regions} of the interference free MACs $M(\{1,2\},1)$ and $M(\{3,4\},2)$ are achievable, as compared to the parameter range where their \emph{sum capacities} are achievable.
\begin{figure}
\centering
\psfragscanon
\psfrag{x}[]{$h_{1}^2P_1$}
\psfrag{y}[b]{$h_{2}^2P_2$}
\psfrag{B1}[]{\tiny{$P_1$}}
\psfrag{B4}[b][b][1][90]{\tiny{$P_2$}}
\psfrag{B2}[]{\tiny{$P_1(1+P_1+P_2)$}}
\psfrag{B0}[b]{\tiny{$0$}}
\psfrag{B5}[b][b][1][90]{\tiny{$P_2(1+P_1+P_2)$}}
\psfrag{B6}[b][b][1][90]{\tiny{$(P_1+P_2)(1+P_1+P_2)$}}
\psfrag{B3}[]{\tiny{$(P_1+P_2)(1+P_1+P_2)$}}
\psfrag{IVSI}[l]{\footnotesize{Individually very strong}}
\psfrag{VSCI}[l]{\footnotesize{Very strong combined}}
\psfrag{MSVS}[l]{\footnotesize{Mixed strong-extremely strong}}
\includegraphics[width=\columnwidth]{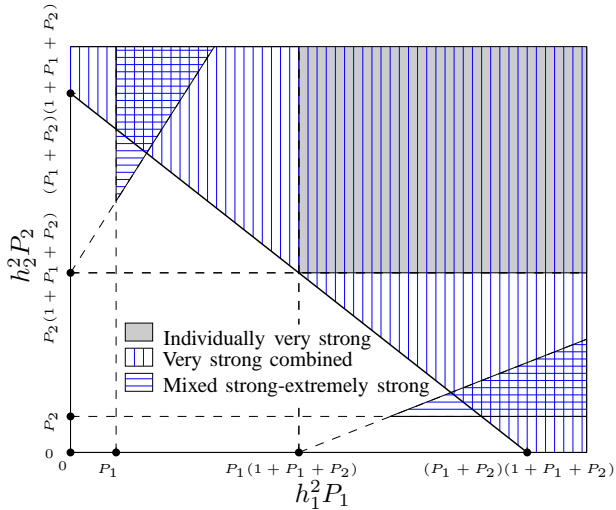}
\caption{In the parameter range within the shaded area (see \eqref{VSI-1}), the capacity region of the interference free MACs $M(\{1,2\},1)$ and $M(\{3,4\},2)$ is achievable simultaneously, while in the range within the vertically dashed area (see~\eqref{VSI-5}), this holds only for the sum capacities. The horizontally dashed area denotes the parameter range with mixed strong--extremely strong interference.}
\label{Thresh}
\end{figure}

\subsection{Sum capacity upper bound for the IMAC}
In this subsection, we provide an upper bound on the sum capacity of the symmetric IMAC. In \cite{AnnapureddyVeeravalli,ShangKramerChen,MotahariKhandani}, a genie aided technique was used to obtain a sum capacity upper bound for the IC that coincides with the simple lower bound of treating interference as noise. Thus the sum capacity of the IC in the so-called noisy interference regime was obtained. We use a similar technique to that used in \cite{AnnapureddyVeeravalli} to obtain an upper bound for the IMAC. This upper bound is stated in the following theorem.

\begin{theorem}
\label{SumCapacityUpperBound}
The sum capacity of the IMAC is upper bounded by $\overline{C}_\Sigma$, \begin{align}
C_\Sigma\leq\overline{C}_\Sigma&\triangleq\min_{\substack{\rho\in[-1,1],\\
\eta^2\leq1-\rho^2}}\log\left(1+\frac{1}{\eta^2}\left(\inr+\frac{\mathsf{A}}{1-\rho^2+\inr}\right)\right)
\end{align}
where $\inr=h_{1}^2P_1+h_{2}^2P_2$, and $$\mathsf{A}=P_1(\eta-\rho h_{1})^2+P_2(\eta-\rho h_{2})^2+P_1P_2(h_{1}-h_{2})^2.$$
\end{theorem}
\begin{proof}[Proof sketch]
We bound $R_1+R_2$ and $R_3+R_4$ by using a genie aided approach similar to \cite{AnnapureddyVeeravalli}. We give receiver 1 the genie signal $S_1^n=h_1X_1^n+h_2X_2^n+\eta_1W_1^n$ and receiver 2 $S_2^n=h_1X_3^n+h_2X_4^n+\eta_2W_2^n$ where $W_i\sim\mathcal{N}(0,1)$ and $\mathbb{E}[W_iZ_i]=\rho_i$, $i=1,2$. After adding the bounds, we observe that their sum is maximized by Gaussian inputs if $\eta_i^2\leq 1-\rho_j^2$, $j\neq i$. By evaluating the upper bound for Gaussian inputs, we obtain the desired expression.
\end{proof}

A sum capacity lower bound is obtained by using Gaussian codes and treating interference as noise, namely
\begin{align}
C_\Sigma\geq\underline{C}_\Sigma&=\log\left(1+\frac{P_1+P_2}{1+h_{1}^2P_1+h_{2}^2P_2}\right).
\end{align}

In Figure \ref{Figure:SumCapacity}, we plot the upper bound $\overline{C}_\Sigma$ and the lower bound $\underline{C}_\Sigma$ for an IMAC with $P_1=P_2=P$, $h_{1}=0.3$, and $h_{2}=0.15$. Notice that this upper bound is nearly tight up to some value of $P$. Intuitively, this means that below some threshold value of $\inr$, treating interference as noise achieves sum rate very close to the sum capacity $C_\Sigma$.

\begin{figure}
\centering
\psfragscanon
\psfrag{UB}[l]{$\overline{C}_\Sigma$}
\psfrag{LB}[l]{$\underline{C}_\Sigma$}
\psfrag{x}[t]{P}
\psfrag{y}[b]{Sum Rate(bits/channel use)}
\includegraphics[width=0.9\columnwidth]{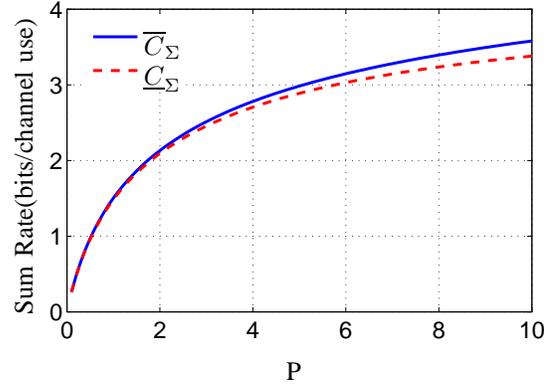}
\caption{Sum Capacity upper bound and lower bound for an IMAC with $P_1=P_2=P$, $h_{1}=0.3$, and $h_{2}=0.15$.}
\label{Figure:SumCapacity}
\end{figure}

Figure \ref{Figure:GapWeakInterference} shows the gap between $\overline{C}_\Sigma$ and $\underline{C}_\Sigma$ in bits per channel use for an IMAC with $P_1=P_2=5$ versus $h_1$ and $h_2$.

\begin{figure}
\centering
\psfragscanon
\psfrag{x}[l]{$h_1$}
\psfrag{y}[l]{$h_2$}
\psfrag{l1}[c]{\scriptsize{$<0.1$}}
\psfrag{l2}[cb]{\scriptsize{$<0.2$}}
\psfrag{l3}[cb]{\scriptsize{$<0.4$}}
\psfrag{l4}[cb]{\scriptsize{$<0.8$}}
\psfrag{l5}[c]{\scriptsize{$<1.6$}}
\includegraphics[width=0.9\columnwidth]{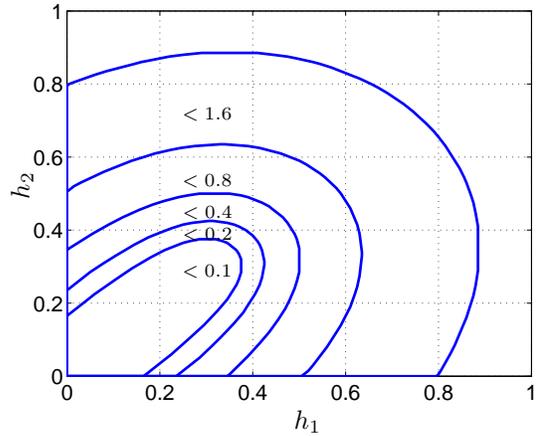}
\caption{This plot shows the gap $\overline{C}_\Sigma-\underline{C}_\Sigma$ as a function of $h_1$ and $h_2$ when $P_1=P_2=5$. Each region denotes a set of pairs $(h_1,h_2)$ where the gap is smaller than the indicated value $0.1$, $0.2$...}
\label{Figure:GapWeakInterference}
\end{figure}

\section{Conclusion}
\label{conclusion}
In this paper, progress has been made towards understanding the interfering MACs (IMAC) channel. The capacity region was characterized under various special cases, namely mixed strong--extremely strong interference, individually very strong interference, very strong combined interference. In the opposite extreme, when the interference is weak, a genie-based upper bound was obtained which is asymptotically tight and nearly tight for reasonably weak cross links.

\bibliography{myBib}

\end{document}